%% file: ISIT24.tex
\documentclass[conference]{IEEEtran}
\usepackage[left=1.62cm,right=1.62cm,top=1.85cm,bottom=2.6cm]{geometry}
\usepackage{float}
\usepackage{amsthm,amssymb,graphicx,multirow,amsmath,color,amsfonts}
\usepackage[update,prepend]{epstopdf}%
\usepackage[noadjust]{cite}
\usepackage{tikz}
\usetikzlibrary{arrows,calc}		
\usepackage{bbm} 
\usepackage{p dfpages}
\usepackage{tabulary}
\usepackage{multirow}
\usepackage{comment}
\usepackage{dsfont}
\usepackage{pgf,tikz}
\usepackage{paralist}
\usepackage{amsfonts}
\usepackage{setspace,lipsum}
\usepackage{url}
\usepackage[T1]{fontenc} 
 \usepackage{algorithmic}
 \usepackage{algorithm}

 \usepackage{thm-restate}
\usepackage{graphicx}
\usepackage{bigints}
\usepackage{mathtools}
\usepackage[english]{babel}
\usepackage[autostyle]{csquotes}
\usepackage{textcomp}
\usepackage{varwidth}
\usepackage{tcolorbox}
\usepackage{lipsum}
 
\let\OLDthebibliography\thebibliography
\renewcommand\thebibliography[1]{
  \OLDthebibliography{#1}
  \setlength{\parskip}{0pt}
  \setlength{\itemsep}{0pt plus 0.3ex}
}

\usepackage{amssymb,amsmath,amsfonts,amsthm}
\usepackage{mathrsfs}
\usepackage{cite}
\usepackage{array}
\usepackage{tabularx}
\usepackage{color}
\usepackage{mathtools}
\usepackage{subfigure}
\usepackage{mathtools}
\usepackage{tikz,pgf}
\usetikzlibrary{calc,positioning,mindmap,trees,decorations.pathreplacing}
\usepackage{graphicx}
\usepackage{bbm}
\usepackage[utf8]{inputenc}
\usepackage{algorithm}

\renewcommand{\IEEEQED}{\IEEEQEDopen}

\include{notation}

\allowdisplaybreaks 

\usepackage{setspace}	

\newcommand{\subparagraph}{} 
\usepackage{titlesec}        

\begin{document}

\author{\IEEEauthorblockN{ \textbf{Ashirwad Sinha}$^{1}$, \textbf{Shubhransh Singhvi}$^{1}$, \textbf{Praful D. Mankar}$^{1}$, \textbf{ Harpreet S. Dhillon}$^{2}$\\}
  \IEEEauthorblockA{$^1$%
  Signal Processing  \&  Communications Research  Center, IIIT Hyderabad, India}
\IEEEauthorblockA{$^2$%
                     Bradley Department of Electrical and Computer Engineering, Virginia Tech, Blacksburg, VA, 24061, USA}
 }
 
\title{Peak Age of Information under Tandem of Queues}
\date{\today}
\maketitle
\thispagestyle{empty}	
\pagestyle{empty}

\begin{abstract}
This paper considers a communication system where a source sends time-sensitive information to its destination via queues in tandem. We assume that the arrival process as well as the service process (of each server) are memoryless, and each of the servers has no buffer. For this setup, we develop a recursive framework to characterize the mean peak age of information (PAoI) under {\em preemptive} and {\em non-preemptive} policies with $N$ servers having different service rates. For the preemptive case, the proposed framework also allows to obtain mean age of information (AoI). 

\end{abstract}
\section{Introduction}
Transfer of time-sensitive information is a vital aspect in many use cases of modern communication networks, such as healthcare, remote actuation, and automated systems. It is often crucial for a destination node  to receive up-to-date information about a physical process observed by  the source node located somewhere else in the network. However, traditional metrics like end-to-end delay may not be sufficient to characterize the performance of such applications, as they do not indicate the time at which the received messages/updates were sampled by the source. To overcome this, a new metric called the {\em age of information} (AoI) has gained popularity in recent years for measuring the freshness of updates received at the destination. The statistical properties of AoI have been extensively studied in the literature for a variety of network settings mostly involving single-hop communication link. However, the information transfer between two nodes often needs to be navigated through multiple gateways placed in the core network along the end-to-end communication link. As these gateways are also responsible for scheduling the traffic of other links,  they may cause random delays in forwarding the messages corresponding to the intended destination node, which in turn affects the  AoI performance.  
Inspired by this, we aim to characterize the AoI  in a multi-hop communication. While doing so, we develop a new recursive approach that offers an appealing alternative to existing approaches that are computationally complex for large settings, as discussed shortly. 

{\em Related works:}
Since the introduction of the AoI metric in \cite{SanjitKaul_2012}, a substantial amount of research has been conducted to characterize and optimize age performance in a variety of communication networks. The mean and distribution of AoI as well as of its variant, known as peak AoI (PAoI), have been extensively analyzed for various queuing systems, for examples see  \cite{Costa_2016,Mohamed_TIT,Soljanin_2017,Kam_2016,Yoshiaki_2019,sinha2022age} for single source case and \cite{Yates_2019,Yates_2018_Networks,Najm_2018,Sun_2018,Mohamed_SHS,Yates_2012,Singhvi_2023} for multi-source  case. The interested readers can refer to \cite{Roy_Survey} and \cite{pappas_abdelmagid_zhou_saad_dhillon_2023} for an excellent  survey on the mainstream  analyses of AoI. In particular, the authors of \cite{Yates_2019} have developed an analytical approach based on the Stochastic Hybrid Systems (SHS) \cite{hespanha2004stochastic} for determining the moments of  age under a system described by finite-state continuous-time Markov chain. The SHS-model has become a main facilitator for  deriving the mean AoI and moment generating function of AoI for complex systems. \\
\indent One of the research directions in the AoI analysis that has received relatively less attention is the characterization of age in multi-hop network settings, which is the focus of this paper. 
The authors of \cite{Talak_2017_multihop} derived the age-optimal scheduling policies for multi-hop systems involving multiple interfering source-destination pairs.
The authors of \cite{Bedewy_2017,Bedewy_2019} investigated the performance of age in multi-hop networks and demonstrated that the age is minimized under a preemptive policy when the multi-hop link  service times follow exponential processes. However, for the case of general distributions of those service times, the authors show that a non-preemptive policy is age-optimal. Further, in \cite{abd2023distribution}, the distribution of age for the gossip network is derived using SHS model.   
Furthermore, by a careful application of SHS-model, the authors of \cite{Yates_2018_Networks, Yates_2020_networks} have derived a closed-form expression of the mean AoI for systems consisting of $N$ {\em preemptive queues} in tandem. Conversely, the authors of \cite{Kam_2018, kam_Frontiers} applied SHS to derive closed-form expressions of the mean AoI under {\em non-preemptive queues} in tandem for $N = 2$ (with different service rates) and $N=3$ (with equal service rates). It was shown that  the SHS model leads to  an intractable set of linear equations  for $N>3$ case.
\emph{To the best of our knowledge, the analysis of mean PAoI for preemptive policy and mean AoI/PAoI for the non-preemptive policy is not present in the literature for the general case of $N$.} 

Further, using SHS model, the authors of \cite{Doncel_2021}   derived a closed-form expression for an upper bound of the mean AoI under  Jackson Networks with finite buffer size.
It is worth noting that the  SHS model can be computationally expensive to analyze or simulate, especially when dealing with large state spaces or systems with complex dynamics/interactions, which makes it challenging to obtain efficient and scalable solutions. Further, it often requires certain simplifying assumptions, such as Markovian  state  space, to make the analysis tractable. Inspired by such limitations of the existing approaches, we develop  a new recursive framework to evaluate the mean performance of PAoI for preemptive and non-preemptive queues in tandem. Our framework also allows to obtain mean AoI for preemptive queues in tandem.

{\em Contributions:}
 This paper focuses on analyzing the mean PAoI under the tandem of queues while assuming memoryless processes for the arrival process and the service process of each server with different rates. For a system with $N$ unit capacity queues, we develop a new recursive analytical framework to evaluate the mean PAoI under the preemptive and non-preemptive policies. Using the recursive framework, we also evaluate the mean AoI for the preemptive policy which was first characterized in \cite{Yates_2018_Networks}. The proposed framework can be utilized to analyze the distribution of PAoI. 
 Overall, it provides an appealing alternative to the existing approaches, such as SHS, that are computationally complex in large-scale settings.

 Our analysis provides useful insights into how the service processes of servers in tandem and how their disciplines impact the  performance of age.\vspace{-1mm}
\section{System Model}\vspace{-1mm}
\label{section:system_model}
We consider a status update system with $N$ servers between the source and destination. Each server is assumed to have a unit capacity and is arranged in tandem as shown in Fig \ref{fig:server_diagram}. The source is assumed to provide time-sensitive updates to the first server regarding some physical process, for which the age needs to be measured at the destination. 
The updates from the source are generated according to a Poisson process with an arrival rate  $\lambda$ and the service times of the $i$-th server are assumed to follow an exponential distribution with parameter $\mu_i$. 
The source can be interpreted as a virtual server that emits/serves the updates according to an exponential inter-departure time at rate $\mu_0=\lambda$. We consider {\em preemptive} and {\em non-preemptive} queuing policies. In the preemptive policy, the in-service updates at each server are replaced with newly arriving updates. On the other hand, in the non-preemptive, newly arriving updates are dropped at each server if it is busy. 
\begin{figure}[h!]
\centering\vspace{-2mm}
\includegraphics[trim={4cm 9.3cm 4cm 15.6cm},clip,width=.45\textwidth]{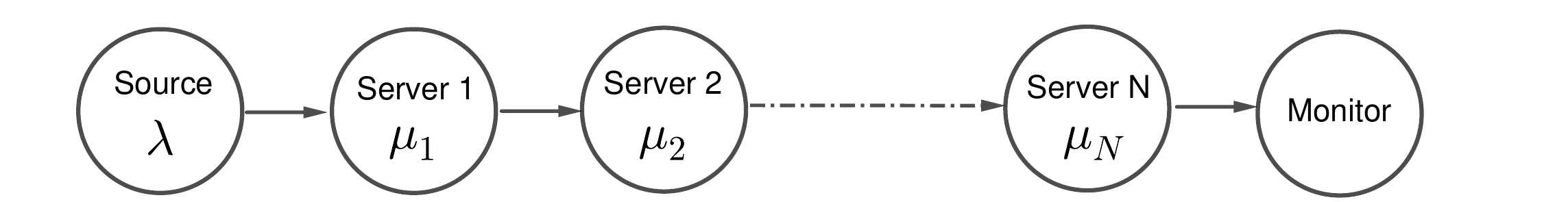} \vspace{-3mm}
    \caption{System model with $N$ queues in tandem. }
    \label{fig:server_diagram}\vspace{-3mm}
\end{figure}

This paper focuses on characterizing  age of updates received by the destination through the tandem of queues, as described above. The AoI at the destination is defined as  
\begin{equation}
    \label{delta(t)}
    \delta(t)=t-U(t),
\end{equation}
where $U(t)$ is the generation instant of the freshest update received by the destination. Fig. \ref{fig:age_sample_path} shows a sample path of AoI under tandem of preemptive queues. As shown in the figure, $\tau_k$ and $\tau^{\prime}_{k,l}$ represent the time of arrival at the first server  and time of departure from the $l$-th server, respectively, of the $k$-th update. Let  $T_k$ = $\tau^{\prime}_{k,n}-\tau_{k}$  be the service time  of the $k$-th update and let $Y_{k} = \tau^{\prime}_{k,n}-\tau^{\prime}_{k-1,n}$ be the inter-departure time between the $k$-th and the $(k-1)$-th updates.
\begin{figure}[h!]
\centering
\includegraphics[trim={1.5cm 3.6cm 1.2cm 4cm},clip,width=.45\textwidth]{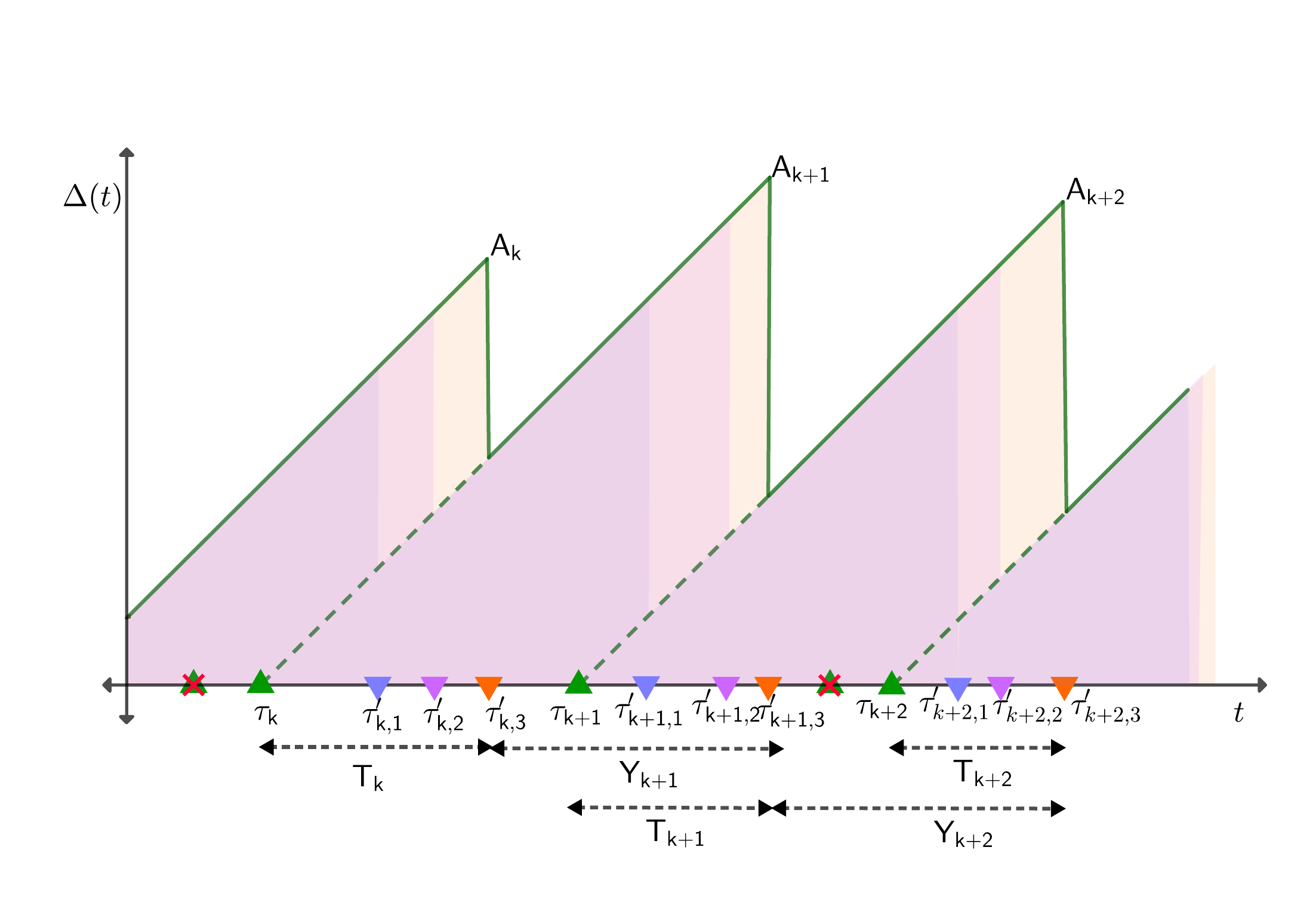} \vspace{-.2cm}
    \caption{Illustration of the sample path of age $\delta(t)$ for $N=3$ under preemptive policy. The green up arrow markers indicate  the arrival instants of updates at server 1 whereas light-blue, magenta, and orange down arrow markers indicate the service instants of the updates from the first, second, and third servers, respectively. The red cross markers represent the  older updates replaced with newer ones.  The three different  shades represent the AoIs observed at the output of  three servers.}
    \label{fig:age_sample_path}
\end{figure}

The mean AoI is defined as the time mean of the age process $\delta(t)$, whereas the mean PAoI is defined as the mean of  $\delta(t)$ observed just before the delivery of updates. The means of  PAoI and  AoI are given by \cite[Equations (8) and (10)]{Costa_2016}
\begin{align}
    \bar{\mathcal{A}}&=\mathbb{E}[Y_k]+\mathbb{E}[T_{k-1}],
    \label{eq:mean_peak_age_defination}\\
    \text{and~}\Delta &= 0.5\frac{\mathbbm{E}[Y_k^2]}{\mathbbm{E}[Y_k]}+\frac{\mathbbm{E}[Y_kT_{k-1}]}{\mathbbm{E}[Y_k]},\label{eq:mean_age_defination}
\end{align}

\section{Age under Preemptive Queues in Tandem}
\label{section:Age_under_Preemptive_Queues_in_Tandem}
In this section, we analyze the age for the case of
preemptive queues in tandem. Let $\psi_i$ represent the event where $i$-th server is busy and all the servers $j>i$ are idle at the time of delivery of the $(k-1)$-th update to the monitor. The analysis of the moments of $T_k$ and $Y_k$ given $\psi_i$ relies on  the successful delivery of the $k$-th update to the monitor. Let us call  this packet as the {\em packet of interest}. In the following, we first analyze the probability of event $\psi_i$. There can be a total of $2^N$ possible states wherein each server can be idle or busy. But for the event $\psi_i$, the state of server $j<i$ is irrelevant as the transition in server $j<i$ will not alter $\psi_i$.  Keeping this in mind, we define our state as $(a,b)$ such that $a>b$ where $a$ is  an index of the server that holds the packet of interest and  $b$ is the index of the server before $a$ that is busy (meaning the servers with index $j$ are idle for $b<j<a$).
The reason for defining the state in such a way is that the only packet in server $b$ can potentially preempt the packet of interest in server $a$. This is because the preemption of the packet in server $b$ will not affect the state $(a,b)$ since the distribution of time required for transition at server $b$ remains unaltered after the preemption. 

Note that upon the delivery of the packet of interest to the monitor from state $(N,b')$, the system enters in state $(N+1,b')$ as the monitor can be interpreted as another server indexed by $N+1$, which acts like a \emph{sink node} that consumes the packets and does not forward them further, i.e., the inter-departure rate at the monitor is $\mu_{N+1} = 0$. Furthermore, the system enters in state $(1,0)$ upon the arrival of the packet of interest at server 1, as the virtual server (i.e., source), indexed by 0, is always busy (since it generates update with inter-departure at rate $\mu_0=\lambda$). Hence, the analysis of $T_k$, $Y_k$, and $\psi_i$ depends on the condition that the packet of interest in state $(a,b)$ is not preempted until it's delivered to the monitor from state $(N,b^\prime)$ for $\forall$  $b^\prime = b, \ldots,N-1$.

In order to derive the probability of $\psi_i$, we first derive the probability of reaching state $(a^\prime,b^\prime)$ from state $(a,b)$.

\begin{lemma}
\label{lemma:one_step_transition_probability}
    The transition probability from state $(a,b)$ to $(a+1,b)$ and $(a,b+1)$ is equal to $\frac{\mu_a}{\mu_a+\mu_b}$ and $\frac{\mu_b}{\mu_a+\mu_b}$, respectively. 
\end{lemma}

\begin{proof}
    Let  $S_k$ be the service time of $k$-th server. The probability of reaching state $(a+1,b)$ from state $(a,b)$ can be simply obtained as $ P(S_a<S_b) = \frac{\mu_a}{\mu_a + \mu_b}$. Similary, the probability of reaching state $(a,b+1)$ from state $(a,b)$ can be obtained as $ P(S_b<S_a) = \frac{\mu_b}{\mu_a + \mu_b}$.
\end{proof}

We denote the \emph{reach probability} from state $(a,b)$ to $(a',b')$ as ${\rm P}(a,b,a^\prime,b^\prime)$, which we obtain in the following lemma. 
\begin{lemma}
\label{lemma:reach_probability}
Reach probability can be obtained recursively as
\begin{align}
\label{eq:reach_probability_recursion}
\begin{split}
        {\rm P}(a,b,a^\prime,b^\prime) =& \frac{\mu_{a}}{\mu_{a}+\mu_{b}}  {\rm P}(a+1,b,a^\prime,b^\prime) \\
    &+ \frac{\mu_{b}}{\mu_{a}+\mu_{b}}{\rm P}(a,b+1,a^\prime,b^\prime),
\end{split}
\end{align}
with the base conditions
\begin{align}
{\rm P}(a,b,a^\prime,b^\prime) = 
\begin{cases}
1, & \text{if~} a=a^\prime \text{~and~} b=b^\prime,\\
0, & \text{if~} a>a^\prime  \text{~or~} b>b^\prime \text{~or~}a\leq b.
\end{cases}
\end{align}

\end{lemma}
\begin{proof}
If  $(a,b)=(a^\prime,b^\prime)$ then our system is already in the final state, and hence ${\rm P}(a,b,a^\prime,b^\prime)=1$. In addition, we have ${\rm P}(a,b,a^\prime,b^\prime)=0$  if  $a> a^\prime$ or $b>b^\prime$ or $a\leq b$ as this transition is not possible by the construction of the state.  

To reach state $(a',b')$ from $(a,b)$ the system must go through either state $(a+1,b)$ or $(a,b+1)$. Therefore, 
\begin{align*}
     {\rm P}(a,b,a^\prime,b^\prime) &= {\rm P}(a,b,a+1,b)\cdot{\rm P}(a+1,b,a^\prime,b^\prime)\\
     &+ {\rm P}(a,b,a,b+1)\cdot{\rm P}(a,b+1,a^\prime,b^\prime).
\end{align*}
 Finally, from Lemma~\ref{lemma:one_step_transition_probability}, we arrive at the result given in \eqref{eq:reach_probability_recursion}. 
    \end{proof}

Using reach probability given in Lemma~\ref{lemma:reach_probability}, we derive the probability of event $\psi_i$ in the following lemma.  
\begin{lemma}
For $0\leq i<N$, 
\begin{equation}
\label{eq:psi_i}
{\rm P}(\psi_i) = \frac{\zeta_i}{\sum_{m=0}^{N-1}\zeta_m},
\end{equation}
where $\zeta_i = {\rm P}(1,0,N,i) \frac{\mu_N}{\mu_N + \mu_i}$. 
\end{lemma}
\begin{proof}
Let $\zeta_i$ denote the probability that the packet of interest is delivered to the monitor from the state $(N,i)$. Therefore,
\begin{equation}
\label{leave_probability}    
\zeta_i = {\rm P}(1,0,N,i){\rm P}(N,i,N+1,i).
 \end{equation}
where ${\rm P}(1,0,N,i)$ can be obtained using Lemma~\ref{lemma:reach_probability}, and from Lemma~\ref{lemma:one_step_transition_probability}, we have ${\rm P}(N,i,N+1,i) =\frac{\mu_N}{\mu_N + \mu_i}$. Further, we obtain the probability of $\psi_i$ simply by normalizing $\zeta_i$. 
\end{proof}

Using the above result, we now obtain the moments of $Y_k$. 
Recall, for a given $\psi_i$,  we have $Y_k = \sum_{j=i}^{N}S_j$ where $S_j$ is the service time of the $j$-th server. 
Since each $S_j$ follows exponential distribution with parameter $\mu_j$ independently of each other, the conditional first and second moments of $Y_k$ for given $\psi_i$ become
    \begin{align}
    \mathbb E[Y_k|\psi_i] &= \sum\nolimits_{m=i}^{N}\frac{1}{\mu_m}, ~~\text{and}\label{eq:yk|psi_i}\\
    \mathbb E[Y_k^{2}|\psi_i] &= \sum_{i\leq l,m\leq N}\frac{2}{\mu_l \mu_m},\label{eq:yk^2|psi_i}
    \end{align}
respectively. Finally, deconditioning \eqref{eq:yk|psi_i} and \eqref{eq:yk^2|psi_i} using \eqref{eq:psi_i} provides the first two moments of $Y_k$ as
    \begin{align}
    \bar{{\rm Y}}_1&= \sum_{i=0}^{N-1}\mathbb{E}[Y_k|\psi_i]{\rm P}[\psi_i],\label{lemma:Y_k_first_moment}\text{~~and}\\
    \bar{{\rm Y}}_2 &= \sum_{i=0}^{N-1}\mathbb{E}[Y_k^2|\psi_i]{\rm P}[\psi_i].\label{lemma:Y_k_second_moment}
    \end{align}

We now focus on determining the mean service time, i.e., $\mathbb E[T_{k-1}]$. For this, let us introduce a useful variable which we call {\em reach time}. Reach time gives the average time to reach state $(a^\prime,b^\prime)$ from the state $(a,b)$ and is denoted as ${\rm T}(a,b,a^\prime,b^\prime)$.   

\begin{lemma}
\label{lemma:reach_time}
    Reach time  ${\rm T}(a,b,a^\prime,b^\prime)$ can be obtained by recursively by evaluating \eqref{eq:reach_time} given at the top of the next page.
\begin{figure*}[!t]
\begin{align}
\label{eq:reach_time}
&{\rm T}(a,b,a^\prime,b^\prime) =  \frac{1}{\mu_{a}+\mu_{b}}+ \frac{p_1}{p_1+p_2} {\rm T}(a+1,b,a^\prime,b^\prime)+ \frac{p_2}{p_1+p_2}{\rm T}(a,b+1,a^\prime,b^\prime)
\end{align}   
 such that 
\begin{align}
\label{eq:reach_time_base_conditions}
{\rm T}(a,b,a^\prime,b^\prime) &=0,  \text{~~if~~}   \begin{cases}
                                a=a^\prime \text{~and~} b=b^\prime,\\
                                p_1=p_2=0,
                                \end{cases}\\
p_1 = \frac{\mu_a}{\mu_a+\mu_b}{\rm P}(a+1,b,a^\prime,b^\prime) &\text{~~and~~} p_2 = \frac{\mu_b}{\mu_a+\mu_b}{\rm P}(a,b+1,a^\prime,b^\prime).
\end{align}
\par\noindent\rule{\textwidth}{0.4pt}
\vspace{-.4cm}
\end{figure*}
\end{lemma}
\begin{proof}
    The base conditions given in \eqref{eq:reach_time_base_conditions} follow from the similar arguments given in the proof of Lemma \ref{lemma:reach_probability}. Let $p_1$ and $p_2$ be the probabilities of reaching the  state $(a^\prime,b^\prime)$  from $(a,b)$ via state $(a+1,b)$ and state $(a,b+1)$, respectively. From Lemma \ref{lemma:reach_probability}, we get
    \begin{align}
        p_1 &= \frac{\mu_a}{\mu_a+\mu_b}{\rm P}(a+1,b,a^\prime,b^\prime),\text{~~and}\label{eq:p_2}\\
        p_2 &= \frac{\mu_b}{\mu_a+\mu_b}{\rm P}(a,b+1,a^\prime,b^\prime).
    \end{align}
    Now, given that the next state is $(a+1,b)$, the density function of the service time $S_a$ becomes
    \begin{align}
        f_{S_a}(S_a|S_b>S_a) &= \frac{\mathbb{P}(S_b>t|S_a=t)f_{S_a}(t)}{\int_0^\infty \mathbb{P}(S_b>t|S_a=t)f_{S_a}(t){\rm d}t},\nonumber\\
        &=(\mu_a+\mu_b)\exp(-(\mu_a+\mu_b)t).
    \end{align}
    Similarly,  given that the next state is $(a,b+1)$, the density function of the service time $S_b$ becomes
    \begin{align}
        f_{S_b}(S_b|S_a>S_b)=(\mu_a+\mu_b)\exp(-(\mu_a+\mu_b)t).
    \end{align}    
    Hence, we get
    \begin{equation}
    \label{eq:reach_time_one_step}
        \mathbb E[S_a|S_a<S_b] = \mathbb E[S_b|S_b<S_a] = \frac{1}{\mu_a+\mu_b}.
    \end{equation}
To reach state $(a',b')$ from $(a,b)$, the packet of interest must go through either state $(a+1,b)$ or $(a,b+1)$. Therefore, 

\begin{align*}
 {\rm T}(a,b,a^\prime,b^\prime) &= \frac{p_1}{p_1+p_2} \left[{\rm T}(a,b,a+1,b) + {\rm T}(a+1,b,a^\prime,b^\prime)\right]\\
 &+ \frac{p_2}{p_1+p_2}\left[{\rm T}(a,b,a,b+1) + {\rm T}(a,b+1,a^\prime,b^\prime)\right].
\end{align*}
From \eqref{eq:reach_time_one_step}, we get ${\rm T}(a,b,a+1,b) = {\rm T}(a,b,a,b+1)  = \frac{1}{\mu_a+\mu_b}$. Substituting these one step reach times in the above expressions, completes the proof. 

    
\end{proof}
In the next lemma, we calculate the expected service time. 
\begin{lemma}
\label{lemma:tk}
Given $\psi_i$, the expected service time is
    \begin{equation}
    \label{eq:tk|psi_i}
        \mathbb E[T_k|\psi_i] = {\rm T}(1,0,N,i) + \frac{1}{\mu_N + \mu_i},
    \end{equation}
    where ${\rm T}(1,0,N,i)$ can be computed from Lemma~\ref{lemma:reach_time}.     
\end{lemma}
\begin{proof}
    Given $\psi_i$, the expected service time is nothing but the reach time between the states $(1,0)$ and $(N+1,i)$ via the state $(N,i)$. Therefore, we get 
\begin{align*}
    \mathbb E[T_k|\psi_i] = {\rm T}(1,0,N,i) + {\rm T}(N,i,N+1,i). 
\end{align*}
Using this and \eqref{eq:reach_time_one_step}, we get \eqref{eq:tk|psi_i} . 
\end{proof}

From \eqref{eq:psi_i} and Lemma~\ref{lemma:tk}, we get the mean of service time as 
\begin{equation}
    \label{lemma:T_k}
       \bar{{\rm T}}_1 = \sum_{i=0}^{N-1}\mathbb E[T_k|\psi_i] \mathbb P(\psi_i).  
\end{equation}

Finally, we can obtain the mean PAoI by combining the above results together as done in the following theorem. 
\begin{thm}
        The mean PAoI $\bar{\mathcal{A}}$ under preemptive policy can be evaluated by substituting \eqref{lemma:Y_k_first_moment} and \eqref{lemma:T_k} in \eqref{eq:mean_peak_age_defination}.
\end{thm}
\begin{restatable}{cor}{TwoSerPreempt}\label{L_th}
For N = 2 servers, the mean PAoI under preemptive policy is  
    \begin{align*}
\bar{\mathcal{A}} &=  \frac{1}{\lambda} + \frac{1}{\mu_1} + \frac{1}{\mu_2} + \frac{1}{\lambda + \mu_1} + \frac{1}{\lambda + \mu_2} + \frac{1}{\mu_1 + \mu_2} \nonumber\\ 
&~~- \frac{2}{\lambda + \mu_1 + \mu_2}. 
\end{align*}
\end{restatable}

\subsection{Mean Age of Information}

In the following, we show that our recursive framework can also be used to compute the mean AoI. For computing mean AoI, as can be seen from \eqref{eq:mean_age_defination}, we require the expectation of product of $T_{k-1}$ and $Y_k$. This expectation is difficult to calculate directly as joint density function of $T_{k-1}$ and $Y_k$ is not known. However,  we can show that these random variables are conditionally independent given $\psi_i$. This can be seen by recalling that $\psi_i$ is the event where the $k$-th packet is in the $i$-th server at the instant of delivery of the $(k-1)$-th packet.

\begin{prop}
   The random variables $T_{k-1}$ and $Y_k$ are conditionally independent given $\psi_i$.
\end{prop}
\begin{proof}
    Given $\psi_i$, the remaining time required for the delivery of the $k$-th packet is equal to $\sum_{j=i}^{N} S_j$, where $S_j$ is the  time taken by the $j$-th server to serve the $k$-th update. Note that even if the packet in the $j$-th server is preempted, the remaining time will be the same because of the memoryless property of service process. This remaining time is nothing but the inter-departure time $Y_k$ which naturally is independent of the service time of the $(k-1)$-th packet. Therefore, we can say that $Y_k = \sum_{j=i}^{N}S_j$ and $T_{k-1}$ are conditionally independent given $\psi_i$. 
\end{proof}
Hence, using the conditional moments of $T_k$ and $Y_k$ for given $\psi_i$, we can obtain the expectation of  $Y_kT_{k-1}$ as
    \begin{equation}
        \label{eq:E[Y_kT_k-1]}
         \bar{{\rm Z}}_1 = \sum_{i=0}^{N-1} \mathbb E[Y_k|\psi_i] \mathbb E[T_{k-1}|\psi_i]  \mathbb P(\psi_i).
    \end{equation} 
Finally, we can obtain the mean AoI by combining the above results together as done in the following corollary. 

\begin{cor}
    The  mean AoI $\Delta$ can be evaluated by substituting
    \eqref{eq:psi_i},
    \eqref{eq:yk|psi_i}, 
    \eqref{lemma:Y_k_first_moment},
    \eqref{lemma:Y_k_second_moment}, \eqref{eq:tk|psi_i} and \eqref{eq:E[Y_kT_k-1]}  in \eqref{eq:mean_age_defination}.
\end{cor}

\section{Age under Non-preemptive Queues in Tandem}
\label{section:Age_under_Non_Preemptive_Queues_in_Tandem}
In this section, we analyze the age for the case of non-preemptive queues in tandem. In the non-preemptive policy, the packet will be dropped if it observes that the immediately next server is busy at the instant when its service completes at the current server. To capture this, we will redefine the state for provisioning our analysis.  
In state $(a,b)$, server $b$ contains the packet of interest and server $a$ is the nearest busy server  from server $b$ such that $a>b$.  
Note the packet of interest can be  potentially dropped by the update being served by its nearest next server. Hence, in order to determine the mean of service time $T_k$, we need to first find the nearest busy server when the packet of interest enters server $1$ ($a$ becomes the nearest occupied server such that the server  $j$ is idle for $a>j>1$). Let  $\theta_i$ denote an event where  server  $i$ is the nearest occupied server when the packet of interest enters server $1$. To derive the probability of  $\theta_i$, we first derive the probability of successful delivery from state $(a,b)$ in the following lemma. Recall that the monitor is considered to be $(N+1)$-th server.
\setcounter{equation}{27}
\begin{figure*}[t!]
\begin{align}
\label{non_preemption:reach_time}
&{\rm T}(a,b) =  \frac{1}{\mu_{a}+\mu_{b}}+ \frac{p_1}{p_1+p_2} {\rm T}(a+1,b)+ \frac{p_2}{p_1+p_2}{\rm T}(a,b+1).
\end{align}   
such that
\vspace{-0.2cm}
\begin{align}
\label{non_premption:reach_time_base_conditions}
{\rm T}(a,b)& =\begin{cases}
            0,\text{~~if~} b=N+1 \text{~or~} p_1=p_2=0,\\
            \frac{1}{\mu_b} + T(a,b+1), \text{~~if~} a=N+1,
        \end{cases}\nonumber 
\end{align}
\begin{align}
p_1=&\frac{\mu_{a}}{\mu_{a}+\mu_{b}}{\rm P}(a+1,b) \text{~~and~~} p_2=\frac{\mu_{b}}{\mu_{a}+\mu_{b}}{\rm P}(a,b+1).\nonumber
\end{align} 
\par\noindent\rule{\textwidth}{0.4pt}\vspace{-.3cm}
\end{figure*}
\setcounter{equation}{22}
 \begin{lemma}
 \label{lemma5}
     The transition probability that an packet successfully goes from the state $(a,b)$ to the state $(N+1,b^\prime)$ can be recursivley obtained as
\begin{equation}
    \label{eq:non_preemption_reccurence}
{\rm P}(a,b) =  \frac{\mu_{a}}{\mu_{a}+\mu_{b}}  {\rm P}(a+1,b)+ \frac{\mu_{b}}{\mu_{a}+\mu_{b}}P(a,b+1),
\end{equation}  
with the base conditions 
\begin{equation}
\label{non_preemption:base_conditions}
{\rm P}(a,b) = \begin{cases}
                                1, & \text{if~} a=N+1\\
                                0, & \text{~or~}a\leq b.
                                \end{cases}
\end{equation}
 \end{lemma}
\begin{proof}
 Here, $a=N+1$ means the packet in server $a$ has been successfully delivered and there are no packet ahead of the update in server $b$ that will result in its drop. Hence, ${\rm P}(a,b)=1$ if $a=N+1$. Since $a>b$,  we have $ {\rm P}(a,b) = 0$ if $a \leq b$. These conditions form the base conditions for the recursive relation. Given the previous state is $(a,b)$, the probability of the current state being $(a+1,b)$ is $\frac{\mu_{a}}{\mu_a+\mu_b}$ and the probability of the current state being $(a,b+1)$ is $\frac{\mu_{b}}{\mu_a+\mu_b}$. Using this along with the  arguments given in the proof of Lemma \ref{lemma:reach_probability}, we find the recursive relation given in \eqref{eq:non_preemption_reccurence}.   
\end{proof}

In the next lemma, we derive the probability of the event $\theta_i$.
\begin{lemma}
    For $2 \leq i \leq N+1$, 
    \begin{equation}
    \label{non_preemption:theta_i}
    P(\theta_i) = \frac{\eta_i}{\sum_{i=2}^{N+1}\eta_i}, 
\end{equation}
where  $\eta_i$ is given in \eqref{eq:eta_i}.
\end{lemma}
\begin{proof}
The analysis of the probability of $\theta_i$  depends on  the state of the system  when the packet of interest arrives in server 1 and the probability of its successful delivery from that state.

The probability of successful delivery from server $j$ without occurring a new arrival is $\frac{\mu_j}{\mu_j+\lambda}$. Besides, the probability of new arrival in state $(i,0)$ is  $\frac{\lambda}{\lambda + \mu_i}$. Using this, we can determine the probability that the update arriving in state $(i,1)$ gets successfully delivered as  
 \begin{equation}
     \eta_i = \prod_{j=2}^{i-1} \frac{\mu_j}{\mu_j + \lambda}  \frac{\lambda}{\lambda+\mu_i}{\rm P}(i,1),\label{eq:eta_i}
 \end{equation}
 where ${\rm P}(i,1)$ can be obtained using Lemma \ref{lemma5} and  $\mu_{N+1}=0$. Further, normalizing $\eta_i$ gives the probability of $\theta_i$. 
\end{proof}

Interestingly, it can be noted that inter-departure $Y_k$'s are equal in distribution under preemption and non-preemption policies which is mainly because the preemption under the memoryless service process essentially replaces the older update with the new one without affecting the remaining service time statistics. Thus, the moments of $Y_k$ can be determined using \eqref{lemma:Y_k_first_moment} and \eqref{lemma:Y_k_second_moment}. Now, in the following, we determine the moments of $T_k$ conditioned on $\theta_i$. For this, we introduce the {\em reach time} for non-preemption as the average time required for the successful delivery of the packet of interest in server $b$ and denote  it as ${\rm T}(a,b)$. 
\begin{lemma}
\label{non_preemption:lemma:reach_time}
    Reach time for state $(a,b)$ can be obtained by recursively evaluating \eqref{non_preemption:reach_time} given at the top of the this page.
\end{lemma}
\begin{proof}
As $b=N+1$  implies that the update is delivered and $a=N+1$  implies that  there is no update that can potentially drop the update in server $b$, we obtain the base conditions for the recurrence relation.  Further, following the steps given in the proof of Lemma \ref{lemma:reach_time}, we obtain the recursive relation for the reach time as given in \eqref{non_preemption:reach_time}.
\end{proof}

Given that the update arrives in state $(i,1)$, its mean service time can be obtained using the reach time as $\mathbb E[T_k|\theta_i] = {\rm T}(i,1)$. Therefore, we obtain the mean of service time $T_k$ as
\setcounter{equation}{28}
\begin{equation}
    \label{non_preemption:t_k}
    \bar{\rm T}_1 = \sum\nolimits_{i=2}^{N+1} \mathbb {\rm T}(i,1){\rm P}(\theta_i),
\end{equation}
where ${\rm P}(\theta_i)$ is given in Lemma \ref{lemma5}.
Using the above results, we can obtain the mean PAoI as given in the following theroem.
\begin{thm}
        The mean PAoI $\bar{\mathcal{A}}$ under non-preemptive policy can be evaluated by substituting \eqref{lemma:Y_k_first_moment} and \eqref{non_preemption:t_k} in \eqref{eq:mean_peak_age_defination}.
\end{thm}
\vspace{-0.2cm}
\begin{restatable}{cor}{TwoSerNonPreempt}\label{L_th}
For N = 2 servers, the mean PAoI under the non-preemptive policy is 
    \begin{align*}
\bar{\mathcal{A}} &=  \frac{1}{\lambda} + \frac{2}{\mu_1} + \frac{2}{\mu_2} + \frac{1}{\mu_1 + \mu_2} - \frac{2}{\lambda + \mu_1 + \mu_2}.
\end{align*}
\end{restatable}

\section{Numerical Results and Discussion}


To compare mean PAoI between the preemptive and non-preemptive policies, it is sufficient to analyse the behaviour of mean service time as the mean inter-departure time is the same for the two policies. In Fig. \ref{fig:result1}, we show the behaviour of mean service time with respect to the arrival rate $\lambda$ for $N=3,~4,~\text{and}~5$ servers. It can be observed that the mean service time drops with $\lambda$ under preemptive case as expected. However, the trend is opposite for the non-preemptive case which might be attributed to the fact of increased packet drop rate at higher $\lambda$ under non-preemption.

\begin{figure}[H]
\centering
\includegraphics[width=.35\textwidth]{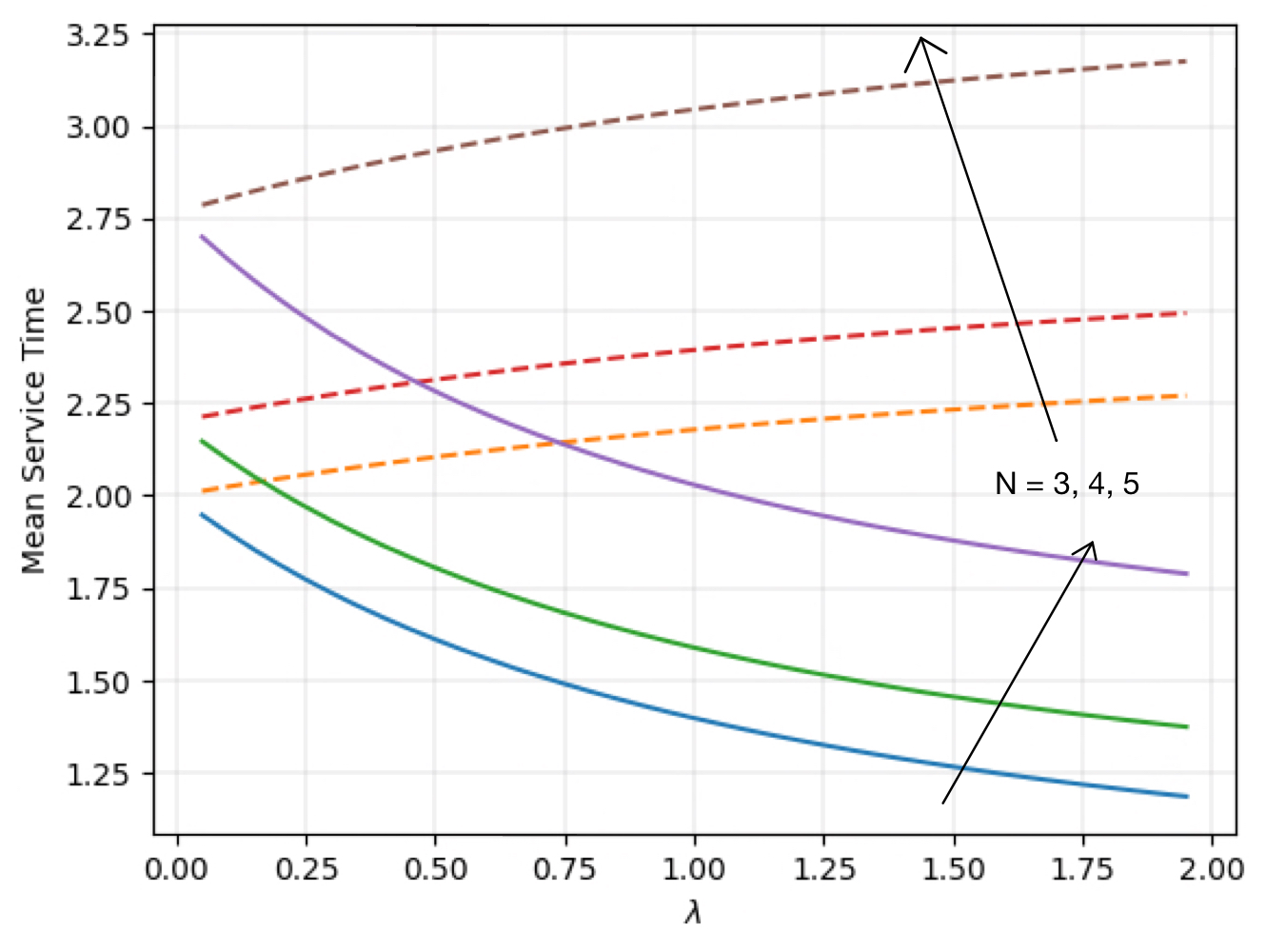}\vspace{-.37cm}
\caption{Mean service time $\bar{\rm T}_1$ vs $\lambda$ for $N=3,4,5$ servers with service times $\mu_i=1.5$ for $i=1,\dots,N-1$ and $\mu_N=1.5,~5,~ \text{and}~10$ for $N=3, 4,~\text{and}~5$, respectively. The solid and dashed lines indicate the mean service times under preemption and non-preemption, respectively.}
\label{fig:result1}\vspace{-2mm}
\end{figure}

\section{Conclusion}
 In this paper, we analyzed the mean peak age performance under the tandem of queues for exponential arrival  and  service processes with different rates. For a system with $N$ unit capacity queues, we developed a recursive analytical framework to evaluate the mean PAoI under the preemptive and non-preemptive policies. We also showed that our recursive framework can be used to compute the mean AoI under the preemptive policy. Our numerical analysis shows that with the increase of arrival rate, the mean service time decreases in preemptive discipline whereas it increases in the non-preemptive policy.
\newpage
\bibliographystyle{IEEEtran} 
\bibliography{ISIT24}

 

\clearpage
\section*{Appendix}

\TwoSerPreempt*
\begin{proof}
Using \eqref{eq:psi_i}, we can evaluate the probabilities of $\psi_o$ and $\psi_1$ as
 \begin{equation*}
 {\rm P}(\psi_0) =  \frac{\mu_1 +\mu_2}{\lambda + \mu_1 + \mu_2}\text{~~and~~} {\rm P}(\psi_1) = \frac{\lambda}{\lambda + \mu_1 + \mu_2}. 
 \end{equation*}
Next, using \eqref{lemma:Y_k_first_moment}, we can the obtain the conditional moments of $Y_k$ as
\begin{align*}
    \mathbb E[Y_k|\psi_0]&=\frac{1}{\lambda}+\frac{1}{\mu_1}+\frac{1}{\mu_2},  \\
    \mathbb E[Y_k|\psi_1]&=\frac{1}{\mu_1}+\frac{1}{\mu_2}.
 \end{align*}
Therefore, the expected inter-departure time is 
\setcounter{equation}{29}
\begin{align}
\label{eq:EIDT_2Ser}
    \overline{Y}_1 = \frac{1}{\lambda}  + \frac{1}{\mu_1} + \frac{1}{\mu_2} - \frac{1}{\lambda+\mu_1+\mu_2}.
\end{align}
 
 Further, using Lemma \eqref{lemma:tk}, we can the obtain the conditional means of $T_k$ as
 \begin{align*}
  \mathbb E[T_k|\psi_0] &=  \frac{1}{\lambda + \mu_1} + \frac{1}{\lambda + \mu_2}, \\   
  \mathbb E[T_k|\psi_1] &=  \frac{1}{\lambda + \mu_1} + \frac{1}{\lambda + \mu_2} + \frac{1}{\mu_1 + \mu_2}.   
 \end{align*}
 
 Plugging the above results together in \eqref{eq:mean_peak_age_defination} and performing some algebraic calculations, we obtain the mean PAoI
\begin{align*}
\bar{\mathcal{A}} &=  \frac{1}{\lambda} + \frac{1}{\mu_1} + \frac{1}{\mu_2} + \frac{1}{\lambda + \mu_1} + \frac{1}{\lambda + \mu_2} + \frac{1}{\mu_1 + \mu_2} \nonumber\\ 
&~~- \frac{2}{\lambda + \mu_1 + \mu_2},
\end{align*}
as desired. 
\end{proof}
\newpage
\TwoSerNonPreempt*
\begin{proof}
Using \eqref{non_preemption:theta_i}, we can evaluate the probabilities of $\theta_2$ and $\theta_3$ as
 \begin{equation*}
 {\rm P}(\theta_2) =  \frac{\lambda}{\lambda + \mu_1 + \mu_2}\text{~~and~~} {\rm P}(\theta_3) = \frac{\mu_1 +\mu_2}{\lambda + \mu_1 + \mu_2}. 
 \end{equation*}

Using Lemma \eqref{non_preemption:lemma:reach_time}, we can the obtain the conditional means of $T_k$ as
 \begin{align*}
  \mathbb E[T_k|\theta_2] &=  \frac{1}{\mu_1} + \frac{1}{\mu_2} + \frac{1}{\mu_1 + \mu_2}, \\
  \mathbb E[T_k|\theta_3] &=  \frac{1}{\mu_1} + \frac{1}{\mu_2}.
 \end{align*}
 
 Plugging the above results together with \eqref{eq:EIDT_2Ser} in \eqref{eq:mean_peak_age_defination} and performing some algebraic calculations, we obtain the mean PAoI
\begin{align*}
\bar{\mathcal{A}} &=  \frac{1}{\lambda} + \frac{2}{\mu_1} + \frac{2}{\mu_2} + \frac{1}{\mu_1 + \mu_2} - \frac{2}{\lambda + \mu_1 + \mu_2}.
\end{align*}
\end{proof}

\end{document}

%% file: notation.tex
\def\nb0{{\mathbf{0}}}
\def\nb1{{\mathbf{1}}}

\newtheorem{lemma}{Lemma}
\newtheorem{thm}{Theorem}

\newtheorem{prop}{Proposition}
\newtheorem{cor}{Corollary}

